\documentclass{article} 

\usepackage{hyperref}
\hypersetup{
     colorlinks   = true,
     citecolor    = blue,
pdfborderstyle={/S/U/W 1}
}
\usepackage{url, tikz}
\usepackage{booktabs}       
\usepackage{amsfonts}       
\usepackage{nicefrac}       
\usepackage{microtype}      

\usepackage{amsmath,amsfonts,amsthm}       
\usepackage{subcaption}
\usepackage{bm}
\usepackage{bbm}
\usepackage{multirow}
\usepackage[inline]{enumitem}
\usepackage{diagbox}
\usepackage[capitalise]{cleveref}  
\usepackage{comment}
\usepackage{etoolbox}
\usepackage{graphicx}
\usepackage{wrapfig}
\usepackage[font=small]{caption}
\usepackage{algorithm,algorithmicx,algpseudocode}
\usepackage{subcaption}

\usepackage{pifont}

\usepackage{import}
\graphicspath{{figures/}{../figures/}}
\usepackage{subfiles}

\newtoggle{comments}
\togglefalse{comments}
\iftoggle{comments}{
    \newcommand\AG[1]{\textcolor{magenta}{[AG: #1]}}
    \newcommand\TD[1]{\textcolor{cyan}{[TD: #1]}}
    \newcommand\KG[1]{\textcolor{blue}{[KG: #1]}}
    \newcommand\IJO[1]{\textcolor{cyan}{[IJ: #1]}}
    \newcommand\AR[1]{\textcolor{red}{[AR: #1]}}
    \newcommand\AT[1]{\textcolor{orange}{[AT: #1]}}
}{
    \newcommand\AG[1]{}
    \newcommand\TD[1]{}
    \newcommand\KG[1]{}
    \newcommand\IJO[1]{}
    \newcommand\AR[1]{}
    \newcommand\AT[1]{}
}

\newtheorem{lemma}{Lemma}[section]

\newtheorem{observation}[lemma]{Observation}

\newtoggle{arxiv}
\toggletrue{arxiv}

\iftoggle{arxiv}{
  \usepackage[square,numbers,sort]{natbib}
  \setlength{\textwidth}{6.5in}
  \setlength{\textheight}{9in}
  \setlength{\oddsidemargin}{0in}
  \setlength{\evensidemargin}{0in}
  \setlength{\topmargin}{-0.5in}
  \newlength{\defbaselineskip}
  \setlength{\defbaselineskip}{\baselineskip}
  \setlength{\marginparwidth}{0.8in}
  \setlength{\parskip}{5pt}%
  \setlength{\parindent}{0pt}%
}{
  \usepackage{imr}
}

\iftoggle{arxiv}{
  \title{Tetrahedralization of a Hexahedral Mesh}
   \usepackage{authblk}
   \author[$\dagger$]{Aman Timalsina}
   \author[$\ddagger$]{Matthew Knepley}
   \affil[$\dagger$]{Department of Computer Science, Purdue University}
   \affil[$\dagger$]{{\texttt{atimalsi@purdue.edu}}}
   \affil[$\ddagger$]{Department of Computer Science and Engineering, University at Buffalo}
   \affil[$\ddagger$]{{\texttt{knepley@buffalo.edu}}}
  \date{}
}{
  \title{Tetrahedralization of a Hexahedral Mesh}
  \author{%
    David S.~Hippocampus\thanks{Use footnote for providing further information
      about author (webpage, alternative address)---\emph{not} for acknowledging
    funding agencies.} \\
    Department of Computer Science\\
    Cranberry-Lemon University\\
    Pittsburgh, PA 15213 \\
    \texttt{hippo@cs.cranberry-lemon.edu} \\
  }
  \bibliographystyle{imr}
}

\begin{document}

\maketitle
\title{\uppercase{}}
\date{
}

\abstract{
Two important classes of three-dimensional elements in computational meshes are hexahedra and tetrahedra. While several efficient methods exist that convert a hexahedral element to tetrahedral elements, the existing algorithms for tetrahedralization of a hexahedral complex limit pre-selection of face divisions. We generalize a procedure for tetrahedralizing triangular prisms to tetrahedralizing cubes, and combine it with certain other triangulation techniques to design an algorithm that can triangulate any hexahedra.}

\iftoggle{arxiv}{}{
}

\thispagestyle{empty}
\pagestyle{empty}

\section{Introduction}

Tetrahedralization of hexahedra has several applications: rendering engines may only process tetrahedra, discretization methods may only require tetrahedra, and some geometric algorithms are only phrased over tetrahedra. Thus, it would be advantageous to convert a hexahedral mesh into as few tetrahedra as possible. Several algorithms exist for this purpose including the popular marching tetrahedra algorithm. These algorithms take advantage of the most natural subdivision of a hexahedron into tetrahedra, and thus, have an inherent simplicity in terms of both understanding and implementation. However, these common algorithms, including the marching tetrahedra, impose severe constraints on the input mesh as they cannot guarantee a conforming division of an arbitrary hexahedral complex, due to non-matching face splits \cite{doi1991efficient}.

This work aims to provide a general algorithm that works on any hexahedral mesh with arbitrary face divisions. The major contribution of this work is a clean and intuitive
formulation of this problem and a generalization of several well-known triangulation algorithms \cite{lee2017subdivisions, edelsbrunner2001geometry} allowing us to triangulate any hexahedra into five or six tetrahedra, except in an exceptional, degenerate case where we use twelve tetrahedra.

\section{Background} 
\subsection{Hexahedral Triangulation}
The decomposition of any polyhedra into other simpler polyhedra has been studied for centuries. Despite its longevity, the problem is difficult and even in the case of decomposition into tetrahedra, it is known that tetrahedralization is \(\textsf{NP}\)-hard \cite{ruppert1989difficulty}. These decomposition problems of arbitrary geometric complexes have yielded a rich body of theoretical results that have provided existence conditions on decompositions and bounds on the minimum number of required tetrahedra \cite{chazelle1989triangulating, shewchuk1998condition}.

On the practical side, several algorithms have been developed to perform these subdivisions, but all of these are oblivious to the orientations of the face splits. This, however, becomes a problem when the orientations do not match and merging affects split orientations. For instance, in the commonly used {\it marching tetrahedra} algorithm \cite{doi1991efficient}, each cube is split into six irregular tetrahedra by cutting the cube in half three times, where this division takes place by cutting diagonally through each of the three pairs of opposing faces. In this way, the resulting tetrahedra all share one of the main diagonals of the cube. An obvious limitation of this algorithm, however, is that the cuts are predetermined: that is, we are restricted to select cuts with matching orientations of opposite pairs.

It is evident that the case of non-matching orientations is significant.  Indeed, the theoretical justification we provide for our algorithm has been explored in several other works. The most notable ones include the use of the region-face-graph or RF-graph to study subdivisions of three-dimensional complexes \cite{meshkat2000generating, sokolov2016hexahedral}. These works were primarily in the context of combining tetrahedral mesh into other polyhedra, and underscore the importance of arbitrary tetrahedral subdivisions. Incidentally, \cite{meshkat2000generating} mentions that when using the common algorithms with a predefined set of face cuts, the associated hexahedral triangulation fails to detect all the potential hexahedra in a tetrahedral mesh, and the percentage of missed potential hexahedra may be significant and even reach 5\% of the overall mesh. The graphs themselves have several nice properties including the fact that the RF graph corresponding to hexahedra-tetrahedra decomposition is planar. Indeed, in the case of hexahedra-tetrahedra decompositions, these representations mostly match the arguments we develop, but the works themselves merely specified these subdivisions and did not explicitly provide an algorithm for generating them. In fact, the case of subdivisions in the case of non-matching orientations was not known.

\subsubsection{Prism Decomposition}\label{subsec: prism}
In order to specify the decomposition of arbitrary hexahedra, we first start by discussing the prism decomposition procedure we employ in the first three of our cases. This decomposition is well-known and we choose the framework specified by \cite{erleben2005adaptive} where they provide an algorithm for triangulating a prism by choosing face cuts carefully. In particular, we first define rising (\(\mathbf{R}\)) and falling (\(\mathbf{F}\)) cuts (\cref{fig:rise-fall}). The cuts (\(\mathbf{R}\)) and (\(\mathbf{F}\)) simply depend on
whether the split edge is rising or falling as we travel along the extruded prism face in a counterclockwise manner. 

\begin{figure}[!h]
\begin{center}
\framebox[5.0in]{
\iftoggle{arxiv}{\includegraphics[scale=0.5]{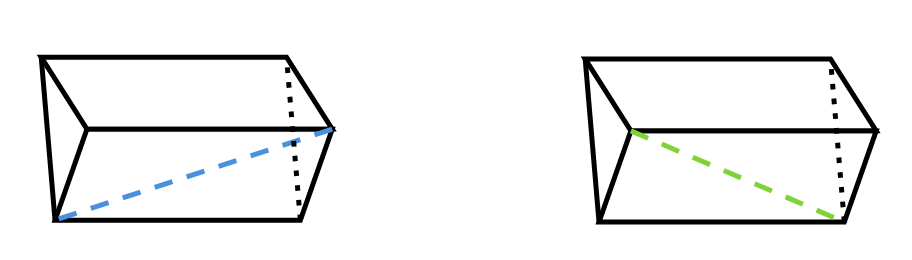}}{
\includegraphics[scale=0.4]{fig/rise-fall-prism.png}}
}
\end{center}
\caption{An instance of (\textit{left}) a rising (\(\mathbf{R}\)) cut and (\textit{right}) a falling (\(\mathbf{F}\)) cut. }
\label{fig:rise-fall}
\end{figure}

\begin{figure}[!h]
\begin{center}
\framebox[5.0in]{
\iftoggle{arxiv}{\includegraphics[scale=0.5]{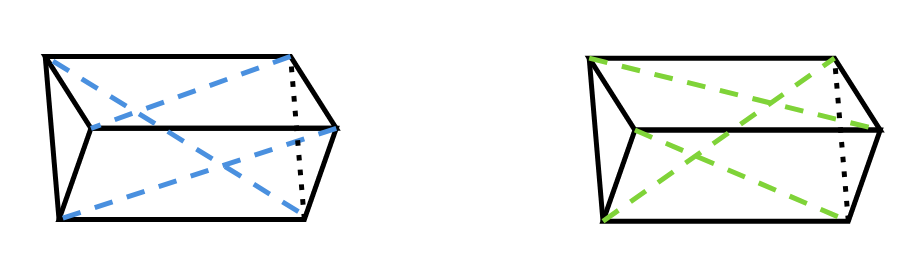}}
{\includegraphics[scale=0.4]{fig/degen-prism.png}}
}
\end{center}
\caption{The degenerate cases: \(\mathbf{RRR}\) cuts (left) and \(\mathbf{FFF}\) cuts (right).}
\label{fig:degen-prism}
\end{figure}
We now claim that the only degenerate cases correspond to instances when all the face cuts are assigned the same orientations, namely \(\mathbf{RRR}\) or \(\mathbf{FFF}\) (\cref{fig:degen-prism}). Note that any other configuration with at least one non-matching orientation guarantees that the face cuts meet at some vertex. Thus, these degenerate configurations definitively characterize the impossibility of triangulating the prism:
\begin{observation}
\label{obs: meet-vertex}
In a tetrahedral decomposition of a prism, it is evident that at least two (exterior) face cuts must meet at some vertex.
\end{observation}
\begin{proof}
We demonstrate this by examining the triangles of the prism. Consider a tetrahedral decomposition $\mathcal{T}$ of a prism $P$. Let $T$ be a tetrahedron in $\mathcal{T}$, and let $u$, $v$, and $x$ be the vertices of the triangular face of $P$ that is adjacent to $T$. Let $s$, the summit, be the vertex of $T$ that is opposite to this face, as shown in Figure \ref{fig: proof}. 
\begin{figure}[!h]
\centering\framebox[5.0in]{
\begin{tikzpicture}
\draw[thick] (0,0) -- (1,1) -- (1,0) -- cycle;
\draw[thick] (1,1) -- (0.5,2.5);
\draw[dashed, thick, red] (1,0) -- (0.5,2.5);
\draw[dashed, thick, blue] (0,0) -- (0.5,2.5);
\draw[fill=black] (0,0) circle (2pt);
\draw[fill=black] (1,1) circle (2pt);
\draw[fill=black] (1,0) circle (2pt);
\draw[fill=black] (0.5,2.5) circle (2pt);
\node[below left] at (0,0.2) {$\bm{u}$};
\node[above] at (1.1,1) {$\bm{x}$};
\node[below right] at (1,0.2) {$\bm{v}$};
\node[above right] at (0.5,2.2) {$\bm{s}$};
\end{tikzpicture}}
\caption{A tetrahedron $T$ adjacent to a triangular face of the prism $P$.}
\label{fig: proof}
\end{figure}
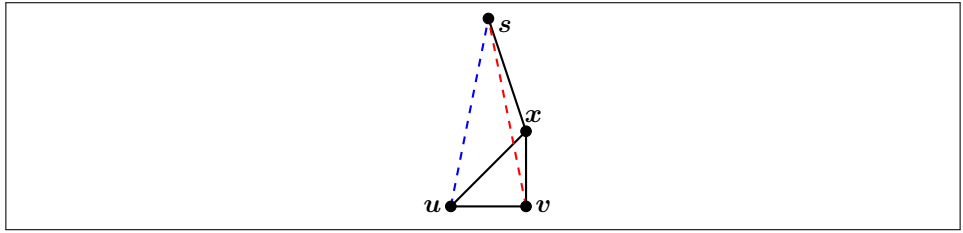
Since $\bm{x}$ can have only three original incident edges, it follows that $\bm{u}$ and $\bm{v}$ must share the three (exterior) face cuts between them corresponding to the three tetrahedra in the triangulation. This means that at least two of these face cuts must meet at a common vertex, which is what the observation states.
\end{proof}

For any of the cuts that are not degenerate (\cref{fig:valid-prism}), a canonical division into three tetrahedra is possible. 
Combinatorially, this yields six different ways of triangulating a prism.

\begin{figure}[!h]
\begin{center}
\framebox[5.0in]{
\iftoggle{arxiv}{\includegraphics[scale=0.5]{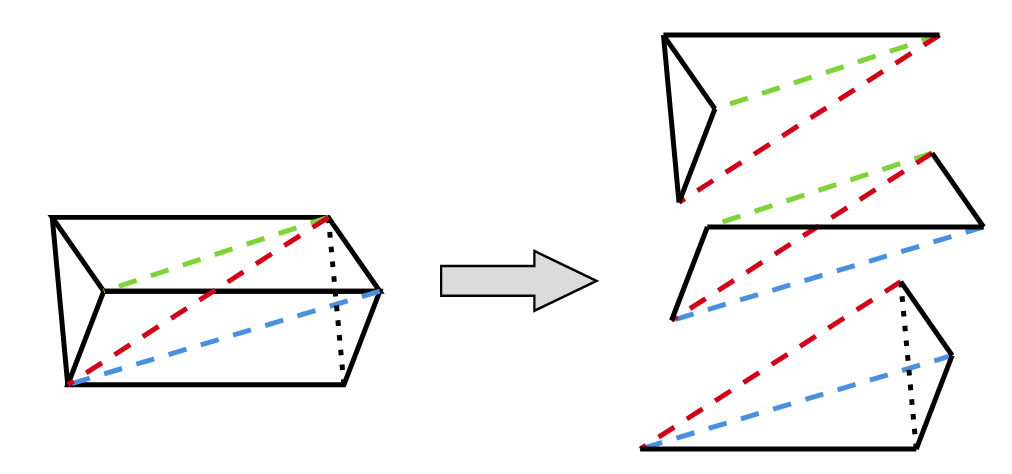}}
{\includegraphics[scale=0.33]{fig/valid-prism.png}}
}
\end{center}
\caption{A valid configuration for prism decomposition.}
\label{fig:valid-prism}
\end{figure}

The issue due to \cref{obs: meet-vertex} is ``fixed'' by \cite{erleben2005adaptive} by looking at the neighboring prisms and changing their configurations to transform these into the non-degenerate cases. We consider a similar strategy for cubes as well, but as we will see, this may not be possible for some global mesh configurations.

\section{General Hex-to-Tet: A General Algorithm for Tetrahedralizing a Hexahedral Complex}
\subsection{Generalizing Prism Decomposition to Cubes}\label{sec: prism-decompose}
Before discussing the generalization of prism decomposition to cubes, we need to clarify some terminology. Recall that the marching tetrahedra algorithm involves partitioning a cube into six irregular tetrahedra by making three cuts along shared diagonals of opposing faces, resulting in the division of the cube into halves three times \cite{doi1991efficient}. We call this shared diagonal the {\it main diagonal}. Further, we extend the notion of rising and falling cuts to cubes as follows. As in the case of prisms, we label the orientation of an external face cut as rising ({\bf R}) or falling ({\bf F}) by traversing along the extruded face in a counterclockwise manner.

Now, a trivial observation that any cube can be divided into two prisms by simply cutting across a diagonal plane allows us to partially reduce arbitrary tetrahedron decomposition of a cube to a decomposition of prisms. Consequently, we can separately triangulate each prism with the main diagonal split serving as a face cut for both prisms. Recall that our main goal was to allow the user to arbitrarily select the face cuts across the six faces; our only freedom being able to choose the main diagonal. We claim that this procedure {\it always} works for cases where up to three cuts have been predetermined. For more than three predetermined cuts, the procedure works if the cuts are lined up accordingly.

For up to three predetermined cuts, we can use the prism decomposition method without running into the degenerate cases from \cref{subsec: prism} as we will always have at least two outside face cuts to choose from. We can choose these cuts in such a way that, along with the main diagonal, these cuts guarantee that we get two cuts that meet at a vertex in each of the prisms, utilizing \cref{obs: meet-vertex}. 

Obviously, solving this case-by-case does not necessarily mean that we obtain a general algorithm for the entire mesh as the cases only correspond to a single hexahedron. However, for simplicity, we will first specify this case-by-case below, and present the main algorithm in \cref{sec:main}. Further, for the sake of exposition, we represent an arbitrary hexahedral element as a cube since they are topologically the same.

\subsubsection{Zero or One Predetermined Cut}
We simply run the marching tetrahedron algorithm here. Or, we can choose an arbitrary main diagonal along with the face cuts in each of the prisms so that prism decomposition can be performed.

\subsubsection{Two Predetermined Cuts}\label{sec: two-pre}
If the two cuts are not opposite to one another, then we can still run the marching tetrahedron algorithm. This is also possible if the two opposite cuts are both falling (\(\mathbf{F}\)) or both rising (\(\mathbf{R}\)). 

Crucially, even in the case where the two cuts have opposite orientations, we can choose the main diagonal so that its end points meet the endpoints of the predetermined cuts (\cref{fig:two}).

\begin{figure}[!h]
\begin{center}
\iftoggle{arxiv}{\framebox[5.0in]{
\includegraphics[scale=0.4]{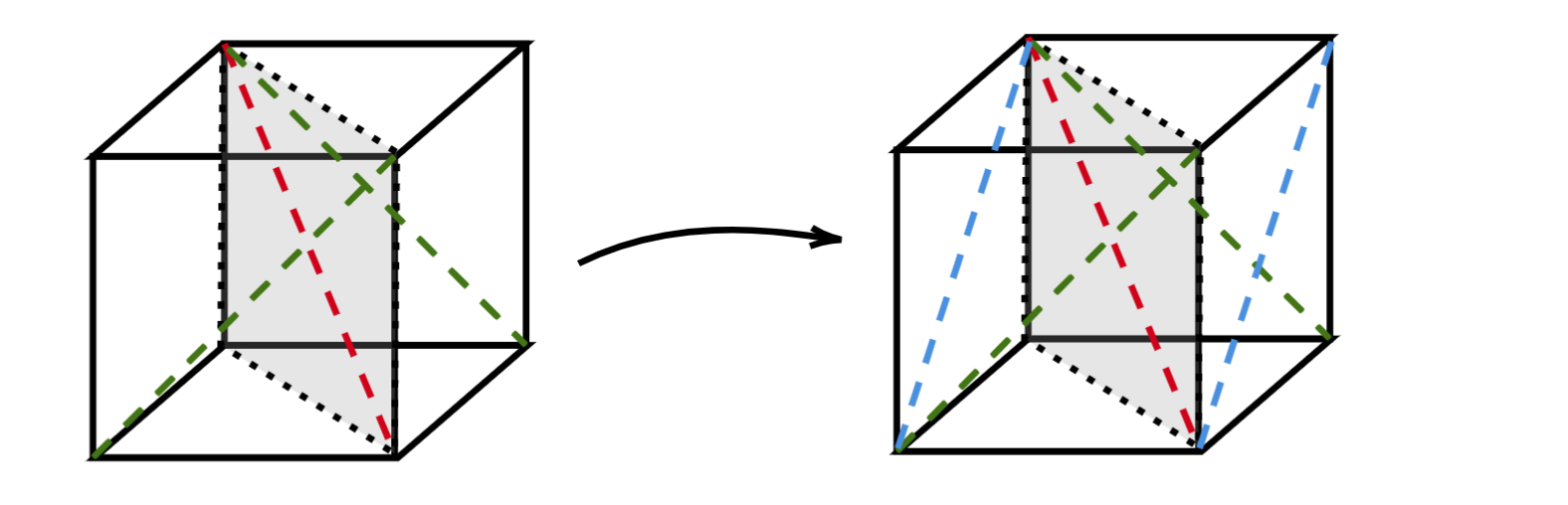}}}{\framebox[3.0in]{
\includegraphics[scale=0.29]{fig/two.png}}
}
\end{center}
\caption{Two predetermined cuts (green): cutting with the opposite faces in the same orientation (blue) yields \(\mathbf{FFR}\) for both prisms.}
\label{fig:two}
\end{figure}

\subsubsection{Three Predetermined Cuts}\label{sec: three-pre}
Combinatorially, the orientation of these face cuts induce the following three cases that we handle separately:

\paragraph{None of the cuts are opposite to one another}
In this case, we should be able to select the opposite cuts for each of the three predetermined cuts. Hence, we can apply marching tetrahedron to get the canonical subdivision.

\paragraph{A pair with opposite face cuts with same orientation}
This is again trivial: we can simply use marching tetrahedron with one of the pairs having already been determined.

\paragraph{A pair with opposite face cuts with {\it different} orientation}
This case employs the following procedure where we want to solve the problem of different orientation by carefully  decomposing the cube into prisms (\cref{fig:three}):

\begin{enumerate}
    \item Choose the uncut pair and decompose the cube into prism by cutting this pair to form the diagonal plane.
    \item On whichever prism the third predetermined cut falls, the main diagonal is cut to avoid the prism degenerate case(s).
    \item The second prism has an uncut external face, which is again used to avoid the prism degenerate case(s).
\end{enumerate}

\begin{figure}[!h]
\begin{center}
\iftoggle{arxiv}{\framebox[5.0in]{
\includegraphics[scale=0.35]{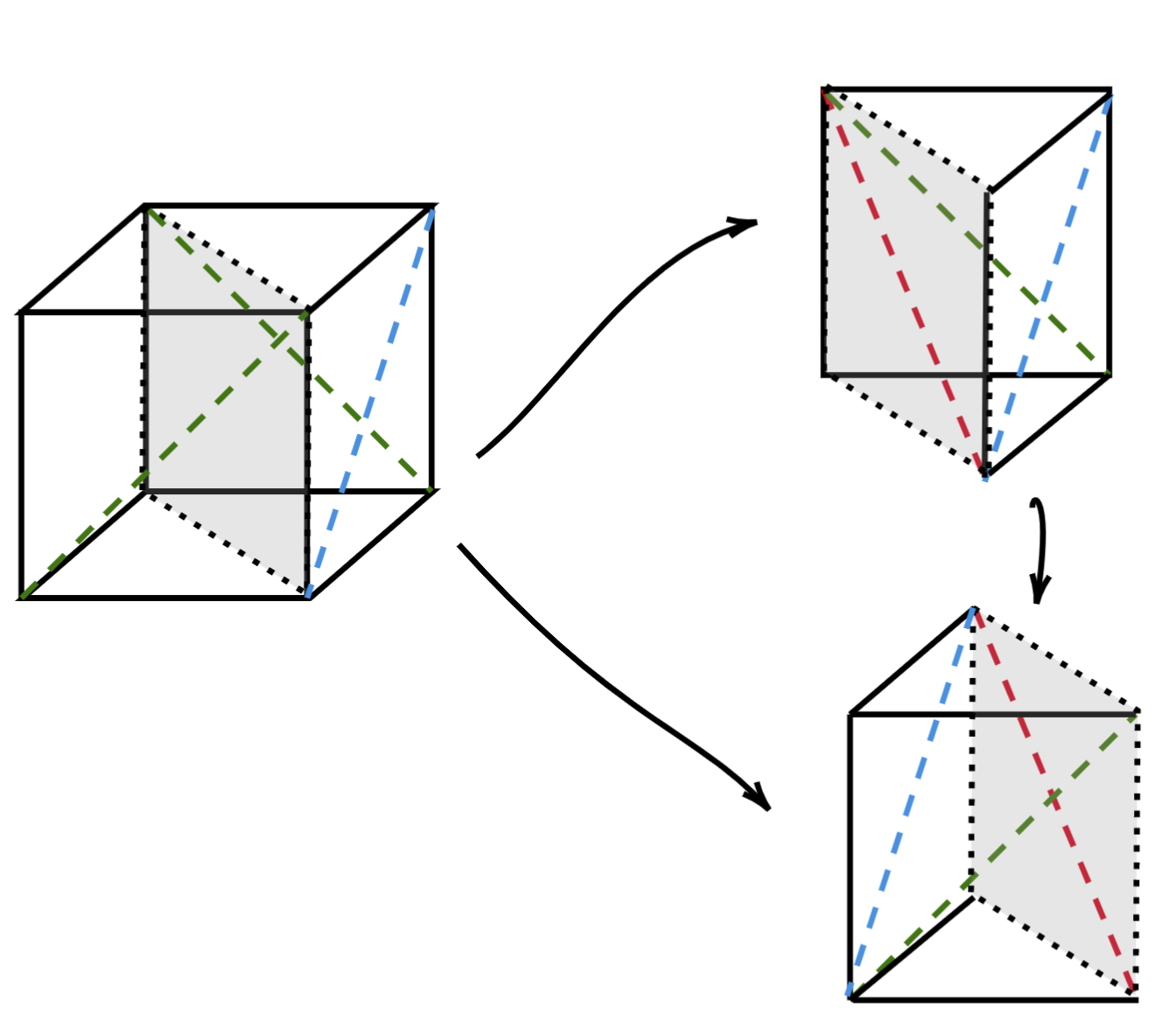}}}
{\framebox[3.0in]{
\includegraphics[scale=0.3]{fig/three.png}}}
\end{center}
\caption{Three predetermined cuts: two of the three cuts form a pair with opposite orientation. ({\it Left}) We first choose a dividing plane. (\textit{Upper-Right}) Then, the red diagonal is chosen
to configure the prism with two external predetermined cuts to \(\mathbf{FFR}\). (\textit{Lower-Right}) Finally, the blue cut is then picked in the second prism
to avoid the degenerate cases.}
\label{fig:three}
\end{figure}

\subsubsection{Four and Five Predetermined Cuts}
Again, we will handle the easy cases first. Also, note that at least one of the pairs of cuts must be opposite to one another here (indeed, two of the pairs must be opposite to one another in the case of five predetermined cuts.)

\paragraph{At least one pair opposite of each other with the same orientation}
This case is easy as we can choose such a pair to create the diagonal plane for decomposing the cube into prisms. If two of the remaining cuts lie on the same prism, then we use the diagonal cut to avoid the degenerate case which leaves us with two (or one) remaining cuts in another prism. If only one cut is present in any one of the prisms, then again we can easily avoid \(\mathbf{RRR/FFF}\) cases. 

\paragraph{Only one pair opposite of each other with different orientation}
We can choose one of the adjacent predetermined cuts and cut across its opposite face with the same orientation. This allows us to choose a diagonal plane with one of the resulting prisms containing two of the predetermined cuts. Here, we can again use the main diagonal cut to escape the degenerate case, while the remaining prism has an extra uncut external face (\cref{fig:one-different}).

\begin{figure}[!h]
\begin{center}
\iftoggle{arxiv}{
\framebox[5.0in]{
\includegraphics[scale=0.3]{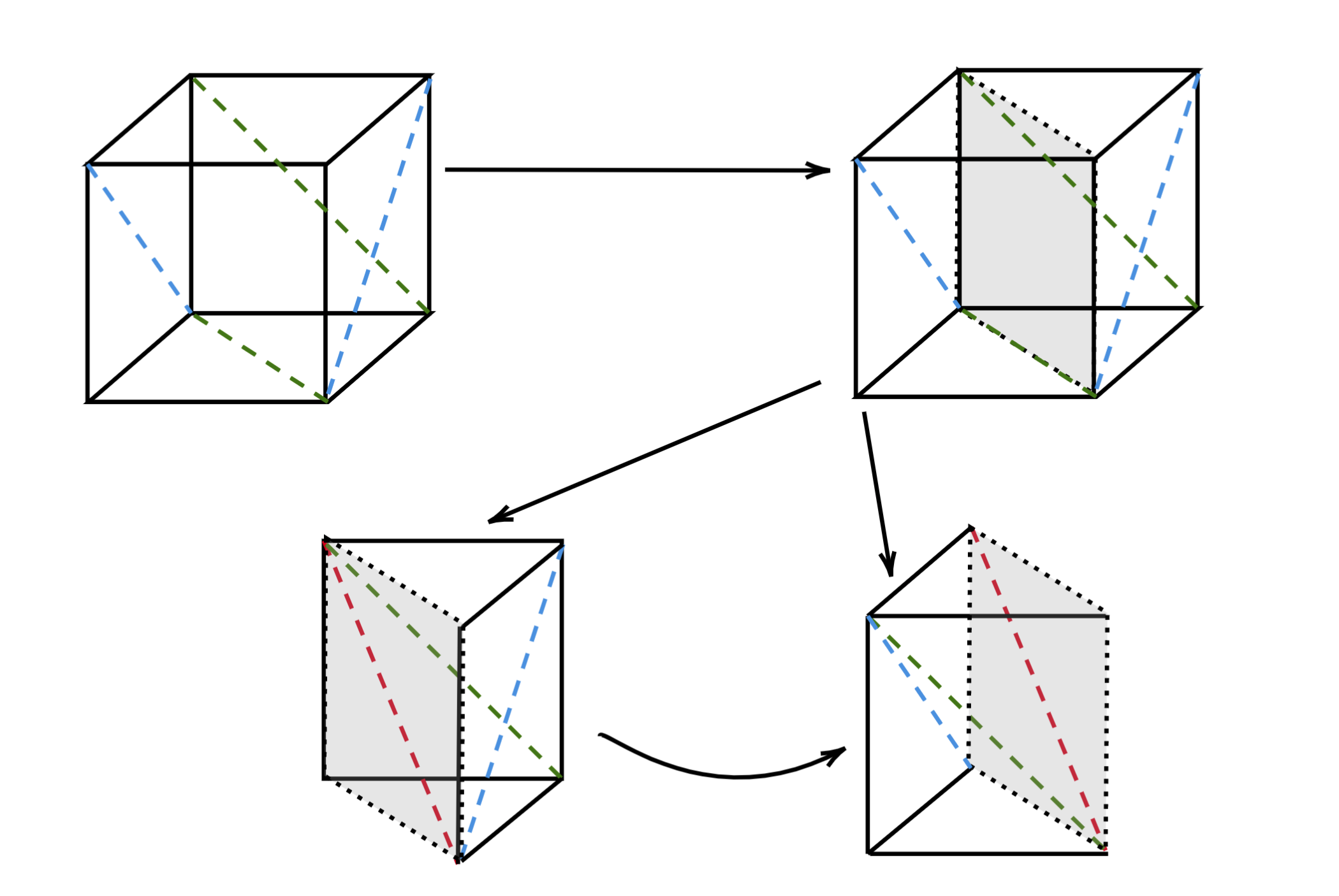}}
}{
\framebox[3.0in]{
\includegraphics[scale=0.2]{fig/one-different.png}}
}

\end{center}
\caption{Decomposition of a single pair with different orientations: (\textit{Top}) The side opposite to the green
cut is chosen so that two of the
predetermined cuts get isolated
in one of the prisms. (\textit{Bottom-Left}) The red diagonal is chosen
to configure the prism to \(\mathbf{FFR}\). (\textit{Bottom-Right}) The green cut is then picked
to avoid the degenerate cases.}
\label{fig:one-different}
\end{figure}

\paragraph{Both pairs are opposite of one another with different orientation}
\subparagraph{Meet at a vertex:} If the pairs meet at a vertex while being in different orientation to their opposite cuts, then we can simply decompose into prisms using the remaining uncut pair (we can cut in the same orientation as the remaining determined cut in the five case), and the resulting prisms are obviously non-degenerate as the external faces meet at a vertex (\cref{fig:four-cuts}).

\begin{figure}[!h]
\begin{center}
\iftoggle{arxiv}{
\framebox[5.0in]{
\includegraphics[scale=0.4]{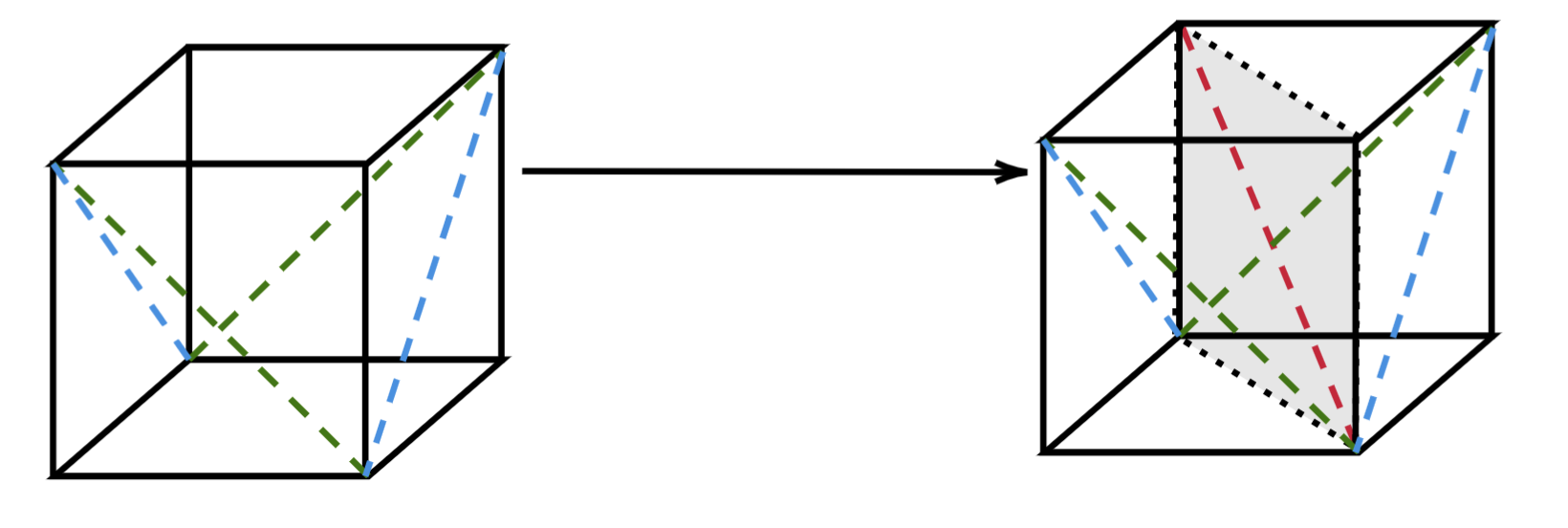}}
}{
\framebox[3.0in]{
\includegraphics[scale=0.27]{fig/four-cuts.png}}
}
\end{center}
\caption{Decomposition of both pairs meeting at a vertex with different orientations: This ensures that we have a pair of opposite uncut faces. (\textit{Right}) Consequently, we can decompose into prisms and choose the red diagonal cut to avoid the degenerate case.}
\label{fig:four-cuts}
\end{figure}

\subparagraph{None meet at a vertex} This is the {\it degenerate case}, where prism decomposition fails. Geometrically, it is equivalent to having a \(\mathbf{RRR/FFF}\) case as in Observation \ref{obs: meet-vertex} for cubes. Note that the remaining one/two cuts cannot save this from being degenerate (\cref{fig:degen-five}).

\begin{figure}[!h]
\begin{center}
\iftoggle{arxiv}{
\framebox[5.0in]{
\includegraphics[scale=0.5]{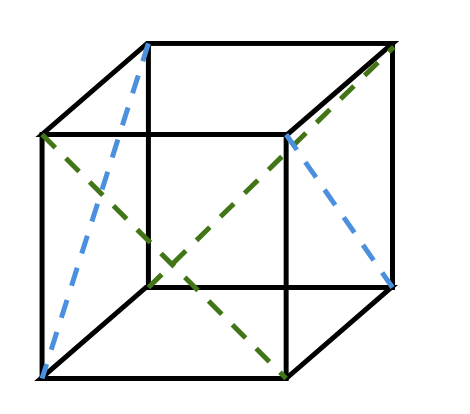}}
}{
\framebox[3.0in]{
\includegraphics[scale=0.4]{fig/degen-five.png}}
}
\end{center}
\caption{Degenerate case for a cube with four determined cuts. The case for degenerate case for cubes with more than four determined cuts is omitted as once four of the exterior face cuts are degenerate as in the figure, the remaining face splits cannot save it from degenerate.}
\label{fig:degen-five}
\end{figure}


\subsubsection{Six Predetermined Cuts}
In this case, any two opposite cuts with the same orientation imply that prism decomposition works. A procedure can be carried out accordingly as in the above cases. 
Again, if any cut is isolated, then the cube cannot be triangulated as Observation \ref{obs: meet-vertex} comes into play.

\subsection{Decomposition into Five Tetrahedra}\label{sec: five-tetrahedra}
Interestingly, even if prism decomposition fails, if at least two opposite pairs have different orientation, then when all pairs meet at (some) vertices to one another, we have the following five tetrahedral decomposition (\cref{fig:five-tets}). We recall that the RF graphs from \cite{meshkat2000generating} also contained cases with decomposition intro five tetrahedra. This is the concrete manifestation of such a decomposition.

\begin{figure}[!h]
\begin{center}
\iftoggle{arxiv}{
\framebox[5.0in]{
\includegraphics[scale=0.3]{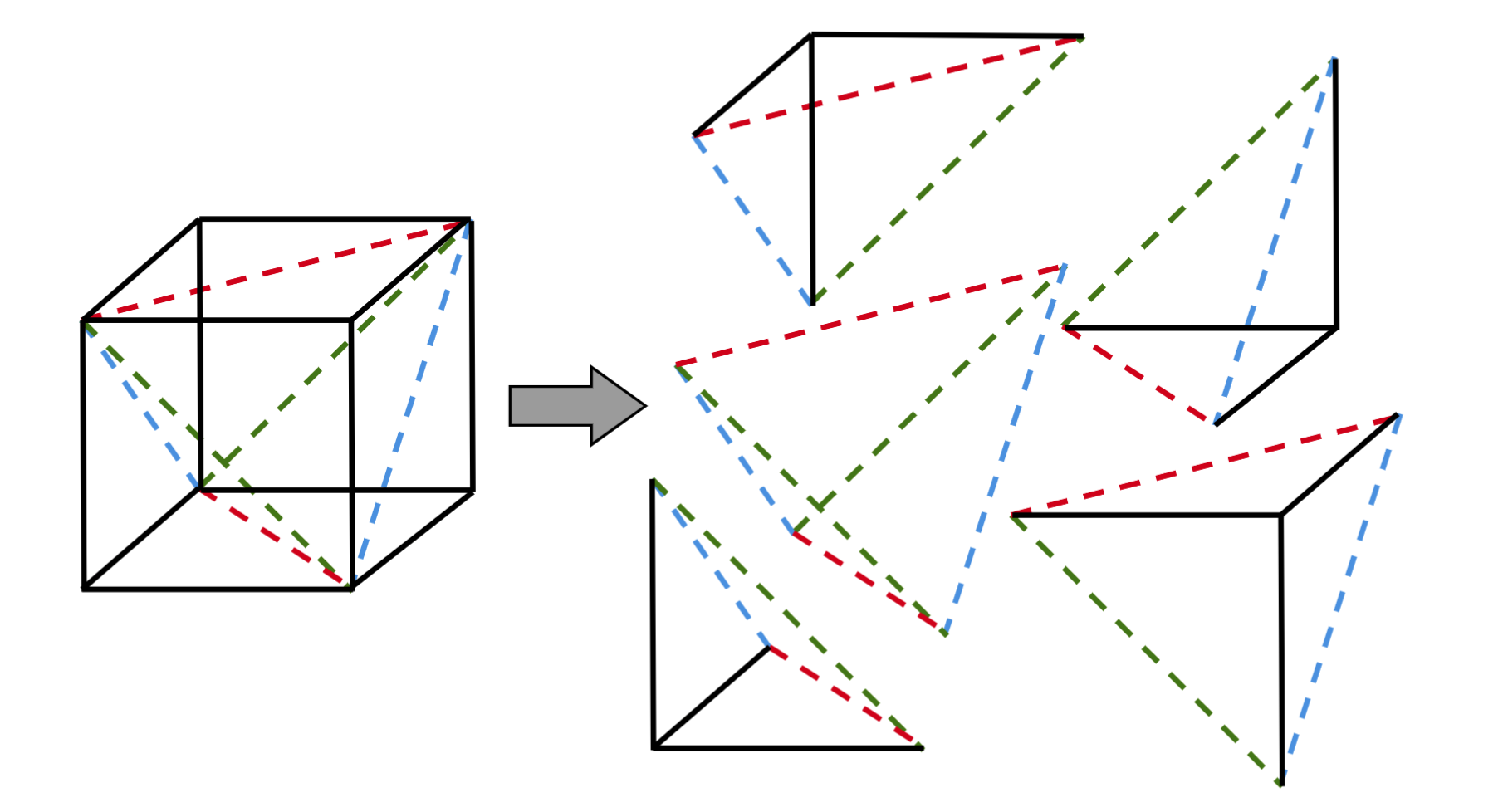}}
}{
\framebox[3.0in]{
\includegraphics[scale=0.26]{fig/five-tets.png}}
}
\end{center}
\caption{The five tetrahedral decomposition: (\textit{Right}) An implementation of this decomposition can be done by simply removing the central tetrahedron; the rest of the four tetrahedra are distributed along the four corners.}
\label{fig:five-tets}
\end{figure}

\subsection{Solving the Degenerate Cases}\label{sec: degenerate}
As we outlined above, the cubes fail to be triangulated in the usual way only if one or more cut(s) are isolated. We claim that this cannot be resolved using one of the procedures above, and one of the following methods must be followed:

\subsubsection{Flipping Neighboring Cubes}\label{sec: flips}
Recall that we had several degrees of freedom when choosing one of the cuts when decomposing and later when avoiding the prism degenerate cases. Indeed, the procedure above is invariant with respect to opposite face cuts of the same orientation. That is, we have the following guarantee on the invariance of flipping cuts of neighboring cubes. 

\begin{observation}\label{thm: opp-adj}
  Changing the orientations of a pair of opposite face cuts still yields a valid tetrahedral decomposition into six tetrahedra.
\end{observation}
\begin{proof}
We want to show that triangulation is invariant under flipping opposite pairs with the same orientation. However, this is essentially changing the diagonal plane that yielded those prisms. 

Assume however that after flipping the orientation, one of the new prism acquires a degenerate configuration (\textbf{RRR/FFF}). In order for this to happen, we must be constrained to cut the main diagonal in some orientation \textbf{R} (respectively \textbf{F}) for one of the prisms with other cuts having orientations \textbf{FF} (respectively \textbf{RR}). The main diagonal will then have orientation \textbf{F} (respectively \textbf{R}), for the other prism, yielding \textbf{FFF} (respectively \textbf{RRR}). 
But, this implies that we started with a degenerate \textbf{FFFF/RRRR} exterior face cuts, which must be impossible as this would not have yielded a valid decomposition before flipping.
\end{proof}

In light of \cref{thm: opp-adj}, in order to avoid the degenerate case for four predetermined case, note that changing the orientation of one of the external face cuts suffices. Here, we note that recursively carrying out this flipping of orientation of face cuts may lead to a "chain reaction" that may end up changing the entire mesh. Thus, we only look for these flips in adjacent cubes and avoid employing this method if the neighbor of the adjacent cube has predetermined face cuts in favor of both simplicity and efficiency.

\begin{figure}[!h]
\begin{center}
\iftoggle{arxiv}{
\framebox[5.0in]{
\includegraphics[scale=0.32]{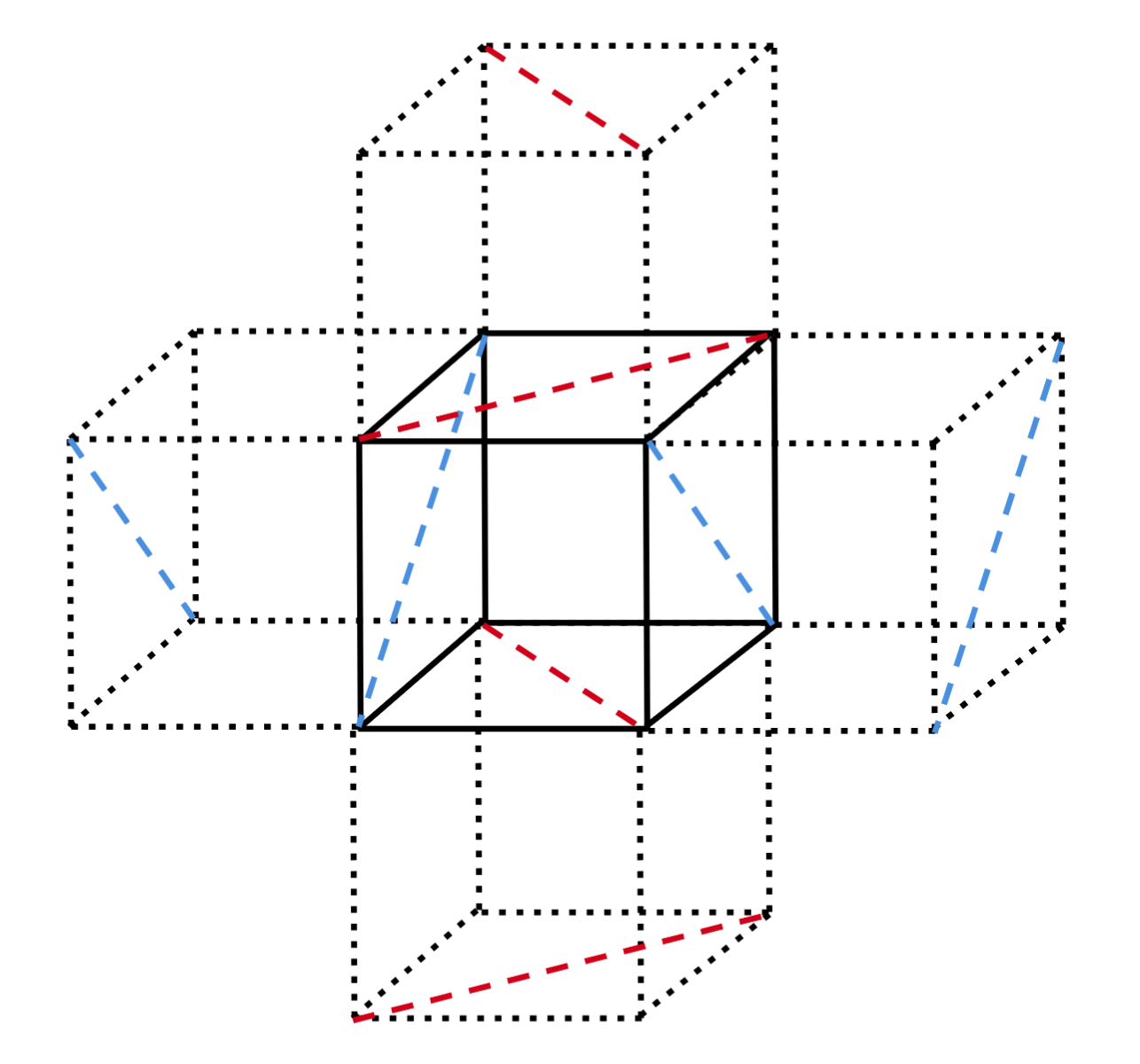}}
}{
\framebox[3.0in]{
\includegraphics[scale=0.27]{fig/flip.png}}
}
\end{center}
\caption{Degenerate case for cubes when flipping opposite face cuts of adjacent cubes: This is an instance when flipping fails as all four adjacent
cubes’ opposite faces have cuts with different orientations.}
\label{fig:flip}
\end{figure}

\subsubsection{Steiner Points}\label{sec: steiner}
Sometimes even flipping fails (\cref{fig:flip}), and our last resort is introducing new vertices, called {\it Steiner points}, which presents an easy solution to the above problem as any number of predetermined cuts can be triangulated to form 12 tetrahedra \cite{de2000computational}: All eight original vertices are connected to the Steiner point to decompose the cube into six pyramids (\cref{fig:steiner}). Any of the face cuts now yields two tetrahedra.

\begin{figure}[!h]
\begin{center}
\iftoggle{arxiv}{
\framebox[5.0in]{
\includegraphics[scale=0.4]{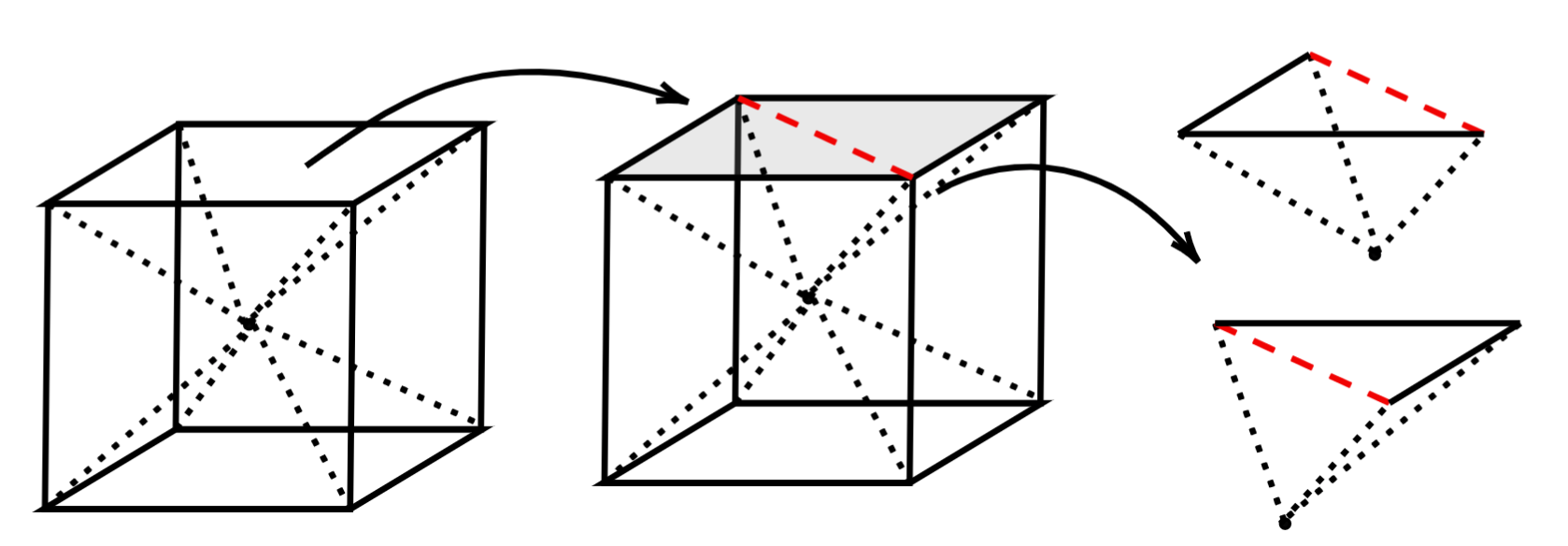}}
}{
\framebox[3.0in]{
\includegraphics[scale=0.26]{fig/steiner.png}}
}
\end{center}
\caption{Steiner points allows decomposition into twelve tetrahedra at the expense of an additional vertex: (\textit{Left}) We first connect each of the vertices
with the Steiner point. (\textit{Middle, Right}) This yields six pyramids, which is then cut across to get
tetrahedra using the predetermined cuts (red).}
\label{fig:steiner}
\end{figure}

\subsection{The Main Algorithm}\label{sec:main}
We have now handled all the cases with certain number of predetermined cuts. We note here that it is quite surprising that the methods used to do this all retain a certain level of simplicity. Indeed, the key idea is knowing the right combination of methods to apply in different cases. However, directly translating these mechanisms into an ad~hoc implementation would not work. Moreover, ignoring the intricacies between the methods may mean that we sacrifice both parallelizability and efficiency in terms of additional vertices. Nevertheless, any attempt at designing an algorithm has to face the degenerate cases; the simplest counterexample is a hexahedralized torus with four cubes where the face cuts are configured in a way that forces the configuration similar to the one from \cref{fig:flip}. 

This counterexample shows that this remains a non-trivial problem to arrange all these cases in a way that allows parallelizability while also ensuring that we use as few Steiner points as possible. We now present a succinct version of such an algorithm that achieves these goals.

   \begin{algorithm}[h]
			\caption{\texttt{Hex-to-Tet (A hexahedral mesh $\mathcal{M}$)}}
			\begin{algorithmic}[1]
			    \While{there exists a hexahedron $H$ that is unmarked}
			        \State $N \leftarrow$ number of exterior cut faces of $H$
			        \If{$N \geq 4$ with two pairs of opposite face cuts with different orientation}
			            \If{the two pairs meet at the same vertices} \Comment{\cref{sec: five-tetrahedra}.}
                                 \While{there exists an uncut face}
			                            \State Cut the face so that the cut meets at the predetermined cut(s) at some vertex 
			                    \EndWhile
                            \Else \Comment{\cref{sec: degenerate}.}
                                \State \texttt{Degenerate-Case ($H$)}
                            \EndIf
			        \Else \Comment{\cref{sec: prism-decompose}}
			            \State \texttt{Prism-Decomposition ($H$)}
			        \EndIf	
			        \State mark $H$
			    \EndWhile
			    \State \Return ``Done''
			\end{algorithmic}
		\end{algorithm}

            \begin{algorithm}[h]
			\caption{\texttt{Prism-Decomposition (Hexahedron $H$)}}
			\begin{algorithmic}[1]
			    \If{there does not exist a pair of opposite cuts with same orientation}
			       \State Cut one of the uncut pairs in this manner
			    \EndIf
			    \State Cut across such pair to create a diagonal plane and two prisms.
			    \If{there exists a prism with two of the exterior faces cut}
			        \State Use the middle diagonal to avoid \(\mathbf{RRR/FFF}\)
			    \EndIf
			    \State Cut the remaining uncut faces of the prism to get valid decompositions
			\end{algorithmic}
		\end{algorithm}
		
	\begin{algorithm}[!h]
			\caption{\texttt{Degenerate-Case (A Hexahedral Mesh $\mathcal{M}$)}}
			\begin{algorithmic}[1]
			    \State Get the adjacent cubes of the four faces with each pair having different opposite orientation
			    \If{any of the four cuts $C$ form an opposite cut pair with the {\it same} orientation in their adjacent cubes} \Comment{\cref{sec: flips}}
			         \State $H' \leftarrow$ neighbor of $H$ that shares the cut $C$
                      \State $C' \leftarrow$ face cut in $H'$ opposite to $C.$
                      \State $H'' \leftarrow$ neighbor of $H'$ that shares the cut $C'$ 
                      \If{$H''$ has {\it not} been marked}
                        \State Flip the orientation of the cut $C'$
                      \EndIf
			    \Else \Comment{\cref{sec: steiner}}
			        \State Introduce a Steiner point $P$ in $H$.
			        \State Link each vertex of $H$ with $P$ using six new interior edges
			        \While{there exists an uncut exterior face}
			             \State Cut the face so that the cut is in the same orientation as its opposite face 
			        \EndWhile
			    \EndIf
			\end{algorithmic}
	\end{algorithm}

\section{Conclusion and Future Work}
\label{sec: conclusion}
In this paper we presented a general triangulation algorithm for hexahedral meshes. Instead of imposing restrictions on the input mesh like other existing algorithms do, our algorithm does not depend on a predefined set of face cuts. Further, our algorithm identifies the number of predetermined face divisions and uses an extension of prism decomposition algorithm and several other techniques to decomposition the hexahedra into tetrahedra. Crucially, we have ensured that our algorithm tries to find all the valid decompositions without making any assumption on the orientations of the face splits, before employing additional vertices. Finally, contrary to previous works, the theoretical framework we inherited extends well to implementation, and in future work, we plan to implement the algorithm above in the PETSc~\cite{petsc-user-ref,petsc-web-page} libraries in order to convert meshes with tensor product cells to simplicial cells as part of its DMPlex mesh capabilities~\cite{KnepleyKarpeev09,LangeMitchellKnepleyGorman2015,KnepleyLangeGorman2017}.	

\paragraph{Acknowledgments}
We gratefully acknowledge the support of the Computational Infrastructure for Geodynamics project NSF EAR-0949446, as well as support from the Department of Energy Applied Math Research program under U.S. DOE Contract DE-AC02-06CH11357.

We would also like to thank the anonymous reviewers for providing helpful comments on earlier drafts of the manuscript. We are particularly indebted to the reviewers for bringing to our attention an error present in the proof of Observation 2.1, which has since been rectified.

\bibliography{main}

\end{document}